\documentclass[aps,pra,showpacs]{revtex4}
\usepackage{epsfig}
\usepackage{amsbsy,latexsym}
\usepackage{amsmath}
\usepackage{amssymb, mathrsfs}
\usepackage[mathscr]{eucal}
\usepackage{hyperref}
\usepackage{amsfonts}
\usepackage{amsthm}
\usepackage{amsmath}
\usepackage{amsfonts}
\usepackage{latexsym}
\usepackage{amssymb}
\usepackage{amscd}
\usepackage[latin1]{inputenc}
\usepackage{verbatim}

\usepackage[english]{babel}

\newtheorem{definition}{Definition}[section]
\newtheorem{theorem}[definition]{Theorem}
\newtheorem{lemma}[definition]{Lemma}
\newtheorem{corollary}[definition]{Corollary}
\newtheorem{proposition}[definition]{Proposition}
\theoremstyle{definition}

\newcommand\style{\mathcal }          

\newcommand{\T}{\style{T}}
\newcommand{\F}{\style{O}}
\newcommand{\B}{\style{B}}
\newcommand{\M}{\style{M}}
\renewcommand{\H}{\style{H}}
\newcommand{\K}{\style K}

\newcommand{\N}{\style{N}}
\newcommand{\X}{\style{X}}

\newcommand{\R}{\style{R}}

\newcommand\tr{ \mbox{\rm Tr} } 

\newcommand\oss{{\style S}}
\newcommand\ost{{\style T}}

\newcommand\spec{{\rm Sp}}      

\newcommand\csta{{\style A}}

\newcommand\cstar{{\rm C}^*}                              

\newcommand\ce{\simeq_{\rm cl} }   
\newcommand\ace{\simeq_{\rm ap-cl} }                         

\newcommand\lec{\ll_{\rm cl}}   
\newcommand\lea{\ll_{\rm ap-cl}}   
\newcommand\leac{\ll_{\rm ac}}   
\newcommand\lef{\ll_{\rm f}}   

\begin{document}
\title{Approximately Clean Quantum Probability Measures}

\author{Douglas Farenick}
\email{douglas.farenick@uregina.ca}
\affiliation{Department of Mathematics and Statistics, University of Regina, Regina, Saskatchewan S4S 0A2, Canada}
 
\author{Remus Floricel}
\email{remus.floricel@uregina.ca}
\affiliation{Department of Mathematics and Statistics, University of Regina, Regina, Saskatchewan S4S 0A2, Canada}
 
\author{Sarah Plosker}
\email{ploskers@brandonu.ca}
\affiliation{Department of Mathematics and Statistics, University of Guelph, Guelph, Ontario N1G 2W1, Canada}
  
\date{May 8, 2013}

\begin{abstract}
 A quantum probability measure--or quantum measurement--is said to be clean if
it cannot be irreversibly connected to any other quantum probability measure via a quantum channel. 
The notion of a clean quantum measure was introduced by Buscemi \emph{et al} in \cite{clean2005} for finite-dimensional Hilbert space, and was studied subsequently by Kahn \cite{kahn2007} and 
Pellonp\"a\"a \cite{pellonpaa2011}. The present paper provides new descriptions of clean quantum probability measures in the case of finite-dimensional Hilbert space.
For Hilbert spaces of infinite dimension, we introduce the notion of ``approximately clean quantum probability measures'' and characterise this property for measures whose range determines
a finite-dimensional operator system.\end{abstract}


\maketitle
\section*{Introduction}

In this paper we investigate the mathematical ramifications of the notion of \emph{cleanness} 
for positive operator valued measures, which was originally introduced and studied by Buscemi \emph{et al} \cite{clean2005}
in connection with the addition of a preprocessing step to the measurement of a quantum system. A related work of Heinonen \cite{heinonen2005},
which appeared at roughly the same time as \cite{clean2005}, addressed the very general issue of optimality in quantum measurements. Subsequent
studies of clean quantum measurements were undertaken by Kahn \cite{kahn2007} and Pellonp{\"a}{\"a} \cite{pellonpaa2011}.
Pellonp{\"a}{\"a} in fact uses a slightly weaker definition than the original definition of clean quantum measurement put forward in \cite{clean2005}.
To explain this, we first recall some of the standard nomenclature in quantum mechanics. 

A quantum system is represented by a Hilbert space $\H$
and a measurement of the system is represented by a positive operator valued probability measure $\nu:\F(X)\rightarrow\B(\H)$, where
$X$ represents a set of measurement outcomes for the system, $\F(X)$ is a $\sigma$-algebra of measurement events, and $\B(\H)$ is the
space of all bounded linear operators acting on $\H$. The corresponding measurement statistics, which capture
the probability that event $E\in\F(X)$ is measured by the apparatus $\nu$ when the system $\H$ is in state $R$,
are given by the real numbers $\tr_\H(R\nu(E))$.
 
Returning to the ideas introduced in \cite{clean2005}, assume that $\nu$ is a measurement of a quantum system
$\H$. Suppose that a preprocessing step is introduced through the use an irreversible quantum channel $\Phi$ that  
maps the states of $\H$ to states in some other system $\H'$ which is to be
measured by $\nu'$. This preprocessing step is represented mathematically in the Heisenberg picture
by $\nu=\Phi^*\circ\nu'$, where $\Phi^*$ is the dual of $\Phi$, and it
transforms the measurement statistics according to the equation $\tr_{\H}\left(R\,\nu(E)\right)=\tr_{\H'}\left(\Phi(R)\nu'(E)\right)$.
The measurement apparatus $\nu'$ is thought to be cleaner than $\nu$ because of the quantum noise 
introduced by the quantum channel $\Phi$. Put differently, $\nu$ is obtained from a cleaner quantum probability measure $\nu'$ by irreversible preprocessing,
a relation that is denoted by $\nu\lec\nu'$. A
\emph{clean quantum measurement} $\nu$ is one in which there is no $\nu'$ from which $\nu$ is obtained irreversibly by preprocessing.

The authors of \cite{clean2005} were interested primarily in the quantum systems represented by Hilbert spaces of finite dimension; however, in the present paper
we aim to develop mathematical techniques that are suited to the case of infinite-dimensional Hilbert space as well. Although there is precedent for this objective in
the works of Heinonen \cite{heinonen2005}, Kahn \cite{kahn2007}, and Pellonp{\"a}{\"a} \cite{pellonpaa2011}, we allow for an even weaker notion of quantum noise,
and this leads us to a new concept: namely, that of an \emph{approximately clean quantum measurement}. There is a corresponding order relation in
which $\nu \lea \nu'$ denotes that $\nu'$ is approximately cleaner than $\nu$. Such a notion is imposed upon us by the mathematical realities of
infinite-dimensional Hilbert spaces. There is an extensive literature on quantum theory in the context of infinite-dimensional spaces; the monographs of
Davies \cite{Davies-book} and Holevo \cite{Holevo-book2,Holevo-book3} treat the fundamentals of this theory.

The Schr\"odinger--Heisenberg duality manifests itself in our work through the fact that $\B(\H)$ 
is the Banach space dual of the space $\T(\H)$ of trace class operators on $\H$. The key maps of study in quantum information theory, namely quantum channels,
are trace-preserving completely positive linear maps $\Phi:\T(\K)\rightarrow\T(\H)$ (for some Hilbert spaces $\H$ and $\K$). Such $\Phi$ represent the evolution of an
open quantum system in the Schr\"odinger picture. The Banach space adjoint of $\Phi$ is a unital completely positive linear map $\Phi^*:\B(\H)\rightarrow\B(\K)$. Furthermore,
because $\Phi^*$ is a dual map, it is necessarily normal. (See Kraus \cite{kraus1971} or Davies \cite{Davies-book} for an elaboration of these facts.) However,
the set of normal maps is not closed in the topology of interest---namely, the point-ultraweak topology---and for this reason we are led to introduce and make use of what we call
an \emph{approximately normal unital completely positive linear map}.

We end this introduction with a short description of the structure of the paper.
Section \ref{to} is a brief overview of the requisite operator theory and Section \ref{qpm} defines the notions of clean and approximately clean
for quantum probability measures, as well as the order relation $\nu \lea \nu'$. The main results are presented in Section \ref{main results}: Theorem \ref{fd1}, which gives an
analytic description of the order relation $\nu \lea \nu'$ for quantum probability measures, and Theorem \ref{fd2}, which determines the structure of approximately clean
quantum probability measures. The paper concludes with Section \ref{S:remarks}, which contains a number of remarks and observations about the main results.
In particular we show that the relations $\lec$ and $\lea$ are distinct in infinite dimensions, and that there are approximately clean quantum measurements that
are not clean. 

Lastly, because we are using Pellonp\"a\"a's definition of clean quantum probability measure, which is more stringent than the definition used in \cite{clean2005} or \cite{kahn2007}, 
our main theorem (Theorem \ref{fd2}) neither implies nor is implied by the results of \cite{clean2005} or \cite{kahn2007}. 

\section{Operator Theory Background}\label{to}


All Hilbert spaces under consideration are assumed to be separable. If $\H$ and $\K$ are Hilbert spaces, then $\B(\H,\K)$ is the Banach space of all 
bounded linear operators $T:\H\rightarrow\K$. With $\K=\H$, denote $\B(\H,\H)$ by $\B(\H)$, and $\T(\H)$ denotes the ideal of
trace-class operators, where the canonical trace
on $\B(\H)$ is denoted by $\tr(\cdot)$ or by $\tr_\H(\cdot)$ if more than one Hilbert space is being considered. 
The von Neumann algebra $\B(\H)$ is the dual space of $\T(\H)$ in the sense that every 
bounded linear functional $\psi$ on $\T(\H)$ has the form $\psi(R)=\tr(AR)$, $R\in\T(\H)$, for some uniquely determined operator
$A\in\B(\H)$. Thus, every bounded linear operator $\Omega:\T(\H)\rightarrow\T(\H)$ admits an adjoint operator $\Omega^{*}:\B(\H)\rightarrow\B(\H)$
such that, for every $A\in\B(\H)$, the equation $\tr(\Omega(R)A)=\tr(R\Omega^*(A))$ holds for all $R\in\T(\H)$.

Considered as the dual of $\T(\H)$, $\B(\H)$ carries a weak*-topology, which in terms of operator topologies coincides with the 
ultraweak topology (or, $\sigma$-weak topology, as it is often called). Namely, a net $\{X_\alpha\}_{\alpha}\subset\B(\H)$ converges ultraweakly to $X\in\B(\H)$ 
if and only if $\tr(RX)=\lim_\alpha\tr(RX_\alpha)$ for every $R\in\T(H)$. A net $\{\phi_\alpha\}_\alpha$ of bounded linear maps $\phi_\alpha:\B(\H)\rightarrow\B(\H)$
converges to a bounded linear map $\phi:\B(\H)\rightarrow\B(\H)$ in the point-ultraweak topology if, for every $X\in\B(\H)$, $\phi(X)$ 
is the ultraweak limit of the net $\{\phi_\alpha(X)\}_\alpha\subset\B(\H)$. Lastly, the set of all bounded linear maps $\phi:\B(\H)\rightarrow\B(\H)$
of norm $\|\phi\|\leq 1$ is compact in the point-ultraweak topology.  

The cone of all positive $Q\in\B(\H)$ is denoted by $\B(\H)_+$. A positive trace-class operator $R\in\T(\H)$ of trace $\tr(R)=1$ is called a density operator.

\begin{definition} If $Q\in\B(\H)_+$, then $\lambda_{\rm min}(Q)$ and $\lambda_{\rm max}(Q)$ are the nonnegative real numbers
\[
\lambda_{\rm min}(Q)\,=\,\min\{\lambda\,:\,\lambda\in\mbox{\rm Sp}(Q)\}\quad\mbox{and}\quad
\lambda_{\rm max}(Q)\,=\,\max\{\lambda\,:\,\lambda\in\mbox{\rm Sp}(Q)\}\,,
\]
where $\mbox{\rm Sp}(Z)$ denotes the spectrum of an operator $Z\in\B(\H)$.
\end{definition}

\begin{definition}\label{os} {\rm (\cite{choi--effros1977,Paulsen-book})}
An \emph{operator system} is a linear (not necessarily norm-closed) subspace $\oss\subset\B(\H)$ with the properties that
$1\in\oss$ and $S^*\in\oss$ for every $S\in\oss$.
\end{definition}

The operator algebra $\B(\H)$ is an operator system, as is every unital C$^*$-subalgebra $\csta\subset\B(\H)$.

\begin{definition} Assume that  $\oss\subset\B(\H)$ and $\ost\subset\B(\K)$ are operator systems acting on Hilbert space $\H$ and $\K$. A linear map $\phi:\oss\rightarrow\ost$
is \emph{completely positive} if, for every $p\in\mathbb N$, the linear map $\phi^{(p)}:\M_p(\oss)\rightarrow\M_p(\ost)$ in which
\[
\phi^{(p)}\left[S_{ij}\right]_{i,j=1}^p\,=\,\left[\phi(S_{ij})\right]_{i,j=1}^p
\]
has the property of mapping the positive cone $\M_p(\oss)_+$ of $\M_p(\oss)$ into the positive cone $\M_p(\ost)_+$ of $\M_p(\ost)$,
where a $p\times p$ matrix $X$ of operators is positive if $X$ is positive as an operator acting on the Hilbert space $\H^{(p)}=\H\oplus\cdots\oplus\H$.
Furthermore, if a completely positive linear map $\phi:\oss\rightarrow\ost$ is such that $\phi(1)=1$, then $\phi$ is \emph{unital} and $\phi$ is called a \emph{ucp map}.
\end{definition}

Completely positive linear maps of $\B(\H)$ or, more generally, of C$^*$-algebras admit a Stinespring decomposition  \cite[Chapter 4]{Paulsen-book},
which is an extremely important tool by which one studies complete positivity. In contrast, there is no Stinespring decomposition for completely positive maps on operator systems that
are not C$^*$-algebras, which adds a degree of difficulty in working with such structures. 

\begin{definition}  {\rm (\cite{Paulsen-book})}
Two operator systems $\oss$ and $\ost$ are \emph{completely order isomorphic} if there is a linear bijection
$\phi:\oss\rightarrow\ost$ such that $\phi$ and $\phi^{-1}$ are completely positive. If, in addition, $\phi(1)=1$, then 
 $\oss$ and $\ost$ are said to be  \emph{unitally completely order isomorphic}
\end{definition}

The following two definitions are standard, but are included here for completeness.

\begin{definition}
A linear map $\Phi:\T(\H)\rightarrow\T(\K)$ is a \emph{quantum channel} if $\tr_\K\circ\Phi=\tr_\H$ and if $\Phi^{(p)}:\M_p(\T(\H))\rightarrow\M_p(\T(\K))$
maps $\M_p(\T(\H))_+$ into $\M_p(\T(\K))_+$ for every $p\in\mathbb N$.
\end{definition}
  
\begin{definition} A completely positive linear map $\phi:\M\rightarrow\N$ of von Neumann algebras is \emph{normal} if, for
every bounded increasing net $\{H_\alpha\}_{\alpha}\subset\M$ of hermitian operators, $\phi(\sup_\alpha H_\alpha)=\sup_\alpha\phi(H_\alpha)$ in $\N$.
\end{definition}

The adjoint $\Phi^*$ of a quantum channel $\Phi:\T(\H)\rightarrow\T(\K)$ is necessarily
a normal ucp map of $\T(\K)^*=\B(\K)$ into $\T(\H)^*=\B(\H)$ \cite[Chapter 2]{Davies-book}. Hence, using Stinespring's theorem and
our knowledge of the irreducible representations of type I factors, the dual $\Phi^*$ 
of a quantum channel $\Phi:\T(\H)\rightarrow\T(\K)$ can be expressed in a Kraus \cite{kraus1971} decomposition:
\[
\Phi^*(S)\,=\,\sum_k A_k^*SA_k,\quad \forall\,S\in\B(\K)\,,
\]
for some countable set $\{A_k\}_{k}\subset\B(\H,\K)$ for which $\sum_kA_k^*A_k=1\in\B(\H)$. 
(Convergence of sums is with respect to the ultraweak topology of $\B(\H)$.) Consequently,
\[
\Phi(R)\,=\,\sum_k A_kRA_k^*,\quad \forall\,R\in\T(\H)\,.
\]

Normality is an important consideration for our study because a unital completely positive linear map $\phi:\B(\K)\rightarrow\B(\H)$
is the adjoint $\phi=\Phi^*$ of some quantum channel $\Phi:\T(\H)\rightarrow\T(\K)$ if and only if $\phi$ is normal. 
(Of course, if $\dim\K$ is finite,
every completely positive linear map $\B(\K)\rightarrow\B(\H)$ is normal.) We shall also have need of ucp maps that are pointwise-ultraweak limits of normal
ucp maps, and we call such maps approximately normal.

\begin{definition}\label{defn:approx normal} A unital completely positive linear map $\phi:\B(\K)\rightarrow\B(\H)$ is \emph{approximately normal} if there exists a net
$\{\phi_\alpha\}_\alpha$ of ucp maps $\phi_\alpha:\B(\K)\rightarrow\B(\H)$ such that each $\phi_\alpha$ is normal and 
$\phi_\alpha\rightarrow\phi$ in the point-ultraweak topology.
\end{definition}

In finite dimensions, ucp = normal ucp = approximately normal ucp, but in infinite dimensions there is a distinction between all three notions. 
It is rather well known that any ucp map $\phi:\B(\K)\rightarrow\B(\H)$ for which $\phi(K)=0$ for every compact operator $K\in\B(\K)$ cannot be normal. 
In Section \ref{S:remarks}
we give an example of an approximately normal ucp map that is not normal.

Variants of Definition \ref{defn:approx normal} may be used for other classes of maps (such a positive rather than completely positive linear maps) or with other topologies. 
But because of the dual relationship between quantum channels and normal ucp maps, we focus only the class of maps and the topology that are natural from the perspective of duality.

\section{Approximately Clean Quantum Probability Measures}\label{qpm}

Throughout $X$ shall denote a nonempty set and $\F(X)$ will denote a $\sigma$-algebra of subsets of $X$. In the language of probability, 
$X$ is a sample space of measurement outcomes and $\F(X)$ is a $\sigma$-algebra of measurement events. 
If $X$ is a locally compact Hausdorff space, then $\F(X)$ is assumed to be the $\sigma$-algebra of Borel sets of $X$. In particular, if
$X$ is a finite set (endowed the discrete topology), then $\F(X)$ is assumed to be the power set of $X$.

\begin{definition} A map $\nu: \F(X) \to \B(\H)$  is a \emph{positive operator valued probability measure (POVM)}, or a \emph{quantum probability measure}, if
\begin{enumerate}
\item[{(i)}] $\nu(E) \in\B(\H)_+$ for every $E \in \F(X)$,
\item[{(ii)}] $\nu(X)=1\in\B(\H)$, and
\item[{(iii)}] for every countable collection $\{E_k\}_{k \in \mathbb N} \subseteq \F(X)$ with $E_j \cap E_k = \emptyset$ for $j \neq k$ we have
\[
\nu\left(\bigcup_{k\in \mathbb N} E_k \right) = \sum_{k \in \mathbb N} \nu(E_k)\,,
\]
where the convergence on the right side of the equation above is with respect to the ultraweak topology of $\B(\H)$. 
\end{enumerate}
In addition:
\begin{enumerate}
\item[{(iv)}] if $\nu(E\cap F)=\nu(E)\nu(F)$ for all $E,F\in\F(X)$, then $\nu$ is a \emph{projective quantum measure}.
\end{enumerate}
\end{definition}

We are particularly interested in an ultraweakly closed vector space $\T_\nu$ determined by the range $\R_\nu$ of a quantum probability measure $\nu$.
 
\begin{definition} If $\nu$ is a quantum probability measure on $(X,\F(X))$, then
the  \emph{range} of $\nu$ is the set 
\[
\R_\nu\,=\,\{\nu(E)\,:\,E\in\F(X)\}\,\subset\,\B(\H)_+,
\]
and the \emph{measurement space of $\nu$} is the vector space
\[
\T_\nu\,=\,\left(\mbox{\rm Span}_{\mathbb C}\,\R_\nu\right)^{\sigma{\rm-wk}}  \,\subset\,\B(\H)\,,
\]
the ultraweak closure of all linear combinations of operators of the form $\nu(E)$, for $E\in\F(X)$.
\end{definition}

By using the ultraweak closure of $\mbox{\rm Span}_{\mathbb C}\,\R_\nu$, 
the operator system $\T_\nu$ becomes a dual operator system in the sense of \cite{Arveson2007}. In particular, the space $\T_\nu$ can be identified with the dual space of its predual  $\T_{\nu,*}$, that is the Banach space of all weak*-continuous linear functionals on $\T_\nu$. 

 \begin{definition} A \emph{measurement basis} for a quantum probability measure $\nu$ is a finite or countably infinite set $\mathcal B_\nu$ of positive operators 
such that
\begin{enumerate}
\item[{\rm (i)}] $\mathcal B_\nu=\{\nu(E)\,:\,E\in\mathcal F_\nu\}$ for some finite or countable family $\mathcal F_\nu\subset\F(X)$ of pairwise disjoint sets, 
\item[{\rm (ii)}] for every $Z\in \T_\nu$ there exists a unique sequence $\{\alpha_{A,Z}\}_{A\in \B_\nu}$ of complex numbers such that $Z=\sum_{A\in \B_\nu }\alpha_{A,Z}A$ in the weak*-topology,
\item[{\rm (iii)}]  for every $A\in \B_\nu$, the coefficient functional $\varphi_A(Z)=\alpha_{A,Z}$, $Z\in \T_\nu$, is a normal positive linear functional.
\end{enumerate}
If $E_0=X\setminus\left(\bigcup_{E\in\mathcal F_\nu}E\right)$, then the operator $A_0=\nu(E_0)$
is called the \emph{basis residual} for $\mathcal B_\nu$;  if $A_0=0$, then  $\mathcal B_\nu$ is said to admit a \emph{trivial basis residual}.
\end{definition}

Note that $1=A_0+\sum_{A\in \B_\nu } A$, if $\mathcal B_\nu$ is measurement basis for $\nu$.

Our primary concern in this paper is with quantum probability measures for which the measurement space $\ost_\nu$
has finite dimension. The general situation requires a different approach, which will be addressed elsewhere.

\begin{proposition}\label{basispartition} If $\omega$
is a projective measurement such that $\T_\omega$ has finite dimension, then $\T_\omega$ has a measurement basis and
every measurement basis for $\T_\omega$ has trivial residual.
\end{proposition}
 
\begin{proof} Assume that $\omega:\F(X)\rightarrow \B(\K)$ is a projective measurement 
and that $\{\omega(F_1),\dots,\omega(F_m)\}$ is a linear basis for $\T_\omega$. Set $Q_j=\omega(F_j)$ for each $j$; thus,
$Q_1,\dots,Q_m$ are linearly independent (pairwise commuting) projections.
Let $E_1=F_1$ and define, iteratively,
\[
E_j \,=\, F_j\setminus \left(\bigcup_{i=1}^{j-1} F_i\right)\,,\quad j=2,\dots,m\,.
\]
The sets $E_1,\dots,E_m$ are pairwise disjoint and nonempty, and therefore the projections $P_1,\dots,P_m$, where
each $P_j=\omega(E_j)$, are nonzero and pairwise orthogonal. Thus, $\{P_1,\dots,P_m\}$ is a set of $m$ linearly independent
operators whose linear span is a subspace of the $m$-dimensional space $\T_\omega$. Hence, $\{P_1,\dots,P_m\}$ is
a measurement basis for $\T_\omega$. Note that $P_0=1-\sum_{j=1}^mP_j$ is orthogonal to each $P_j$, and thus $P_0$
is either zero or is linearly independent of $P_1,\dots,P_m$. But $m=\mbox{dim}\T_\omega$ yields $P_0=0$,
showing that the measurement basis $\B_\omega=\{P_1,\dots,P_m\}$ for $\omega$ has trivial residual. Indeed this latter argument shows that 
every measurement basis for $\T_\omega$ has trivial residual.
\end{proof}

The decomposition of $X$ into a finite disjoint-union of measurable sets $E_0,\dots, E_m$ in the proof of Proposition \ref{basispartition} is an idea that
appears in other works on quantum measurement--for example, as in the decomposition
of phase space into cells given
in \cite[(10.19)]{Bengtsson--Zyczkowski-book}.

Pellonp\"a\"a's definition of clean quantum probability measure is stated below.

\begin{definition}\label{clean order} {\rm (\cite{pellonpaa2011})}
Assume that $\nu_1$ and $ \nu_2$ are quantum probability measures on $(X,\F(X))$ with values in $\B(\H_1)$ and $\B(\H_2)$
respectively.
\begin{enumerate}
\item $\nu_1$ is \emph{cleaner than} $\nu_2$, denoted by
$\nu_2 \ll_{\rm cl} \nu_1$, if $\nu_2=\Phi^*\circ\nu_1$ for some quantum
channel $\Phi:\T(\H_2)\rightarrow\T(\H_1)$.
\item $\nu_1$ and $\nu_2$ are \emph{cleanly equivalent}, denoted by $\nu_2 \ce \nu_1$, if
$\nu_1 \ll_{\rm cl} \nu_2$ and $\nu_2 \ll_{\rm cl} \nu_1$.
\item $\nu_1$ is \emph{clean} if $\nu_2\ll_{\rm cl}\nu_1$ for every quantum probability measure $\nu_2$ satisfying $\nu_1\ll_{\rm cl}\nu_2$.
\end{enumerate}
\end{definition}

As noted in the Introduction, if $\nu_2\lec\nu_1$ via a quantum channel $\Phi$, then the measurement statistics satisfy
\[
\tr_{\H_2}\left(R\,\nu_2(E)\right)\,=\,\tr_{\H_1}\left(\Phi(R)\nu_1(E)\right),\quad\forall\,R\in\T(\H_2),\;E\in\F(X)\,.
\]

As in classical probability, another of the most basic of partial orderings is the partial order that arises from absolute continuity.

\begin{definition} If $\nu_j:\F(X)\rightarrow\B(\H_j)$ is a quantum probability measure on $(X,\F(X))$ for $j=1,2$, then 
$\nu_2$ is \emph{absolutely continuous with respect to} $\nu_1$,
denoted by $\nu_2\leac\nu_1$, if $\nu_2(E)=0$ for all $E\in\F(X)$ for which $\nu_1(E)=0$.
\end{definition}

There is a Radon-Nikod\'ym derivative that arises from the ordering $\nu_2\leac\nu_1$; see \cite{farenick--plosker--smith2011} for further information.
Pellonp\"a\"a's result below shows that for projective measures the clean ordering is the same as the ordering by absolute continuity.

\begin{theorem}\label{p thm} {\rm (Pellonp\"a\"a \cite[Theorem 3]{pellonpaa2011})} Every projective measurement is clean and
the following statements are equivalent for quantum probability measures $\omega$ and $\nu$ on 
$(X,\F(X))$, where $\omega$ is projective:
\begin{enumerate}
\item $\nu\leac\omega$;
\item $\nu \ll_{\rm cl} \omega$.
\end{enumerate}
\end{theorem}

We now introduce a new form of dominance for quantum probability measures.

\begin{definition}\label{approx clean order} 
Assume that $\nu_1$ and $ \nu_2$ are quantum probability measures on $(X,\F(X))$ with values in $\B(\H_1)$ and $\B(\H_2)$
respectively.
\begin{enumerate}
\item $\nu_1$ is \emph{approximately cleaner} than $\nu_2$, denoted by
$\nu_2 \lea \nu_1$, if $\nu_2=\phi\circ\nu_1$ for some approximately normal ucp map $\phi:\B(\H_1)\rightarrow\B(\H_2)$.
\item $\nu_1$ and $\nu_2$ are \emph{approximately cleanly equivalent}, denoted by $\nu_2 \ace \nu_1$, if
$\nu_1 \lea \nu_2$ and $\nu_2\lea \nu_1$.
\item $\nu_1$ is \emph{approximately clean} if $\nu_2\lea\nu_1$ for every quantum probability measure $\nu_2$ satisfying $\nu_1\lea\nu_2$.
\end{enumerate}
\end{definition}

If $\nu_2 \lea \nu_1$ via an approximately normal ucp map $\phi$, then the corresponding measurement statistics are related by
\[
\tr_{\H_2}\left(R\,\nu_2(E)\right)\,=\,\lim_\alpha\tr_{\H_1}\left(\Phi_\alpha(R)\nu_1(E)\right)\,\quad\forall\,R\in\T(\H_2),\;E\in\F(X)\,,
\]
where $\{\Phi_\alpha\}_\alpha$ is a net of quantum channels $\T(\H_2)\rightarrow\T(\H_1)$
for which $\Phi_\alpha^*\rightarrow\phi$ in the point-ultraweak topology.

\section{Structure of Approximately Clean Quantum Probability Measures with Finite-Dimensional Measurement Spaces}\label{main results}

Our first objective is to characterise the order relation $\nu_2 \lea \nu_1$ in the case where the measurement space $\T_{\nu_1}$ has finite dimension.

\begin{theorem}\label{fd1} If $\B_\nu=\{\nu(E_1),\dots,\nu(E_m)\}$ 
is a finite measurement basis for a quantum probability measure $\nu:\F(X)\rightarrow\B(\H)$, 
then the following statements are equivalent for a quantum probability measure $\nu':\F(X)\rightarrow\B(\H')$:
\begin{enumerate}
\item\label{fd1-1} $\nu' \lea \nu$;
\item\label{fd1-2} $\nu'=\psi\circ\nu$ for some ucp map $\psi:\B(\H)\rightarrow\B(\H')$;
\item\label{fd1-3} for all $L_0,\dots,L_m\in\M_p(\mathbb C)$ and every $p\in\mathbb N$,
\begin{equation}\label{basic ineq}
\left\|  \sum_{j=0}^m \nu'(E_j)\otimes L_j\right\| \leq \left\|  \sum_{j=0}^m \nu(E_j)\otimes L_j\right\|\,.
\end{equation}
\end{enumerate}
\end{theorem}

\begin{proof}
\eqref{fd1-1} $\Rightarrow$ \eqref{fd1-2}. 
If $\nu' \lea \nu$, then $\nu'=\phi\circ\nu$ for some approximately normal ucp map $\phi$. Take $\psi=\phi$.

\eqref{fd1-2} $\Rightarrow$ \eqref{fd1-3}. Because $\nu'=\psi\circ\nu$ for some ucp map $\psi:\B(\H)\rightarrow\B(\H')$, and
because ucp maps
are completely contractive, 
$\left\|  \sum_{j=0}^m \psi(\nu(E_j))\otimes L_j\right\| \leq \left\|  \sum_{j=0}^m \nu(E_j)\otimes L_j\right\|$ for all 
$L_0,\dots,L_m\in\M_p(\mathbb C)$ and every $p\in\mathbb N$.
That is, inequality \eqref{basic ineq} holds.

\eqref{fd1-3} $\Rightarrow$ \eqref{fd1-1}. Assume that 
for all $L_0,\dots,L_m\in\M_p(\mathbb C)$ and every $p\in\mathbb N$, inequality \eqref{basic ineq} holds.
Let $\H^{(\infty)}=\H\otimes\ell^2(\mathbb N)$ and $\H'{}^{(\infty)}=\H'\otimes\ell^2(\mathbb N)$. Because $\H^{(\infty)}$ and 
$\H'{}^{(\infty)}$ are infinite-dimensional separable Hilbert spaces,
they are isomorphic and we may identify them; that is, we assume $\H^{(\infty)}=\H'{}^{(\infty)}$ without loss of generality.
Let $A_j=\nu(E_j)$ and $B_j=\nu'(E_j)$, for $j=0,\dots,m$,
and let $\pi:\B(\H)\rightarrow\B(\H^{(\infty)})$ be given by $\pi(Z)=Z\otimes 1$.  Define $\pi':\B(\H')\rightarrow \B(\H^{(\infty)})$ analogously.
Thus, $\pi$ and $\pi'$ a normal, unital homomorphisms. Set $Z^{(\infty)}=\pi(Z)$ and $Y^{(\infty)}=\pi'(Y)$, for every $Z\in\B(\H)$ and $Y\in\B(\H')$, 
and note that inequality \eqref{basic ineq} implies
\begin{equation}\label{basic2}
\left\|  \sum_{j=0}^m B_j^{(\infty)}\otimes L_j\right\| \leq \left\|   \sum_{j=0}^m A_j^{(\infty)}\otimes L_j\right\|\,,
\end{equation}
for all $L_0,\dots,L_m\in\M_p(\mathbb C)$ and every $p\in\mathbb N$.

Consider now the operator systems $\X_\nu=\mbox{Span}\,\{A_1^{(\infty)},\dots, A_m^{(\infty)}\}$ and 
$\X_{\nu'}=\mbox{Span}\,\{B_1^{(\infty)},\dots, B_m^{(\infty)}\}$, and let
$\psi_0:\X_\nu\rightarrow\X_{\nu'}$ denote the linear function defined by 
\[
\psi_0\left( \sum_{j=1}^m \alpha_j A_j^{(\infty)}\right)\,=\,\sum_{j=1}^m\alpha_j B_j^{(\infty)}\,,
\]
for $\alpha_1,\dots,\alpha_m\in\mathbb C$. 
Inequality \eqref{basic2} shows that the unital linear map $\psi_0$ is completely contractive. Hence, 
by Arveson's extension theorem \cite{Paulsen-book}, $\psi_0$ 
has a completely contractive extension to a ucp map $\psi:\B(\H^{(\infty)})\rightarrow\B(\H^{(\infty)})$, which maps 
the operator system $\X_\nu$ into $\X_{\nu'}$. 

Let $\csta=\cstar(\X_\nu)$, which is a separable unital C$^*$-algebra, and let $\theta=\psi_{\vert\csta}$.
Because the C$^*$-algebra $\csta$ contains no nonzero compact operators, Voiculescu's theorem \cite[II.5]{Davidson-book} 
implies that there is a sequence of isometries $V_n\in \B(\H^{(\infty)})$ such that
$\lim_n\|\theta(T)-V_n^*TV_n\|=0$ for every $T\in\csta$. If $\{\omega_j\}_j$ and $\{e_k\}_{k\in\mathbb N}$ 
are orthonormal bases of $\H$ and $\ell^2(\mathbb N)$, then the
map $V:\H'\rightarrow\H^{(\infty)}$ in which $V\omega_j=\omega_j\otimes e_1$, for all $j$, is an isometry. 
Thus, for each $n$, the map $\phi_n:\B(\H)\rightarrow\B(\H)$ given by
$\phi_n(Z)=V^*V_n^*Z^{(\infty)}V_nV$ is a normal ucp map such that $\lim_n\|B_j-\phi_n(A_j)\|=0$ for all $j=1,\dots,m$.
 Since the space of all ucp maps $\B(\H)\rightarrow\B(\H')$ is compact in the point-ultraweak topology, 
the sequence $\{\phi_n\}_{n}$ admits a cluster point $\phi$, and so there is a subsequence, 
which for notational simplicity we denote again by $\{\phi_n\}_{n}$, such that $\phi(Z)$ is the ultraweak limit of $\{\phi_n(Z)\}_n$ for every $Z\in\B(\H)$, and such that 
$B_j=\phi(A_j)$ for each $j=1,\dots,m$.
\end{proof}

Our second objective is to characterise approximately clean quantum probability measures that have finite-dimensional measurement spaces.
To do so, we begin with a useful preliminary result (Lemma \ref{upper bound}) related to the inequalities in statement \eqref{fd1-3} of Theorem \ref{fd1} above.

\begin{definition} If $A_1,\dots, A_m$ are pairwise commuting positive operators such that $\sum_{j=1}^m A_j=1$, then the
\emph{joint spectrum} of the $m$-tuple $(A_1,\dots,A_m)$ is the subset
$\spec(A_1,\dots,A_m)\subset\mathbb R^m$ 
of all $\lambda\in\mathbb R^m$ for which there exists a unital homomorphism $\rho:\cstar(A_1,\dots,A_m)\rightarrow\mathbb C$
such that $\lambda_j=\rho(A_j)$ for all
$j=1,\dots,m$, where $\cstar(A_1,\dots,A_m)$ is the unital abelian C$^*$-algebra generated by $A_1,\dots,A_m$.
\end{definition}

\begin{lemma}\label{upper bound} If $A_1,\dots,A_m\in\B(\H)$ are positive operators such that $\sum_{j=1}^mA_j=1$, then for every Hilbert space $\K$
and operators $L_1,\dots, L_m\in\B(\K)$ we have
\[
\left\| \sum_{j=1}^m A_j\otimes L_j \right\|\,\leq\,\max_{1\leq j\leq m}\|L_j\|\,.
\]
If, moreover, $A_1,\dots, A_m$ are 
pairwise commuting with joint spectrum $\Lambda\in\mathbb R^m$, then
\[
\left\| \sum_{j=1}^m A_j\otimes L_j \right\|\,=\,\displaystyle\max_{\lambda\in \Lambda } \left\| \displaystyle\sum_{j=1}^m\lambda_jL_j\right\| \,.
\]
In particular, if $A_1,\dots, A_m$ are nonzero
projections, then 
\[
\left\| \sum_{j=1}^m A_j\otimes L_j \right\|\,=\,\displaystyle\max_{1\leq j\leq m}\|L_j\|\,.
\]
\end{lemma}

\begin{proof} Let $G_j=(A_j^{1/2}\otimes 1_\K)\in\B(\H\otimes\K)$; thus,
$G_1,\dots,G_m\in\B(\H\otimes\K)$ are positive operators such that  $\sum_jG_j^2=1\in B(\H\otimes\K)$. 
Pass to matrices of operators: let $L=(1_\H\otimes L_1)\oplus\cdots\oplus (1_\H\otimes L_m)$ and
\[
G\,=\,\left[ \begin{array}{cccc} G_1 &0 &\cdots& 0 \\ G_2 &0 &\cdots& 0 \\ \vdots & \vdots & \cdots & \vdots \\
G_m &0 &\cdots& 0
\end{array}\right]\,.
\]
Thus, $\displaystyle\sum_{j=1}^m A_j\otimes L_j = \displaystyle\sum_{j=1}^m G_j(1_\H\otimes L_j)G_j$ and
\[
\left\| \sum_{j=1}^m G_j(1_\H\otimes L_j)G_j\right\| \,=\,\|G^*LG\|\,\leq\,\|G\|^2\|L\|\,=\,\|G^*G\|\,\max_{1\leq j\leq m}\|1_\H\otimes L_j\|\,.
\]
Because $\|G^*G\|=
\left\| \sum_{j=1}^mG^2_j\right\| =1$ and $\|1_\H\otimes L_j\|=\|L_j\|$, the 
desired general inequality follows.

Suppose now that $A_1,\dots, A_m$ are pairwise commuting. Then the maximal ideal space $\mathfrak M$ of the abelian C$^*$-algebra 
$\cstar(A_1,\dots,A_m)$ is homeomorphic to $\Lambda$ under the map $\rho\mapsto (\rho(A_1),\dots,\rho(A_m))$. Hence, every operator $A_j$
is a complex-valued continuous function on $\mathfrak M$ whereby $\rho\mapsto\rho(A_j)$. 
Because the C$^*$-algebra $C(\mathfrak M)\otimes\B(\K)$ is isometrically isomorphic to the C$^*$-algebra of continuous functions on $\mathfrak M$ with values in $\B(\K)$, 
we deduce that
\[
\left\|\sum_{j=1}A_j\otimes L_j\right\|\,=\,\max_{\rho\in\mathfrak M}\left\| \sum_{j=1}\rho(A_j)\otimes L_j\right\|
\,=\, \displaystyle\max_{\lambda\in \Lambda } \left\| \displaystyle\sum_{j=1}^m\lambda_jL_j\right\| \,.
\]

Lastly, if $A_1,\dots, A_m$ are nonzero projections, then the equation $\sum_{j=1}^mA_j=1$ implies that $A_1,\dots, A_m$ are pairwise orthogonal and, hence, pairwise commuting.
Fix $i$ and choose any $\rho$ in the maximal ideal space $\mathfrak M$ of $\cstar(A_1,\dots,A_m)$ for which $\rho(A_i)=1$. (Such a $\rho$ exists because the spectrum of $A_i$
contains $1$.) Then $1=A_i+\sum_{j\neq i}A_j$ implies that $\rho(1)=\rho(A_i)+\sum_{j\neq i}\rho(A_j)$, and so $\sum_{j\neq i}\rho(A_j)=0$. As each $\rho(A_j)\geq 0$, we obtain
$\rho(A_j)=0$ for every $j\neq i$. Hence, the joint spectrum $\Lambda$ of $(A_1,\dots,A_m)$ contains only the canonical coordinate vectors $e_1,\dots, e_m$ of $\mathbb R^m$
and, therefore, 
\[
\max_{\lambda\in \Lambda } \| \displaystyle\sum_{j=1}^m\lambda_jL_j\|
\,=\,\max_{1\leq j\leq m}\|L_j\|\,,
\]
which completes the proof.
\end{proof}

The following theorem is the main result of the present paper.

\begin{theorem}\label{fd2} The following statements are equivalent for a quantum probability measure $\nu$ such that the measurement space $\ost_\nu$ has finite dimension:
\begin{enumerate}
\item\label{ap-cl-1} $\nu$ is approximately clean;
\item\label{ap-cl-2} there are a Hilbert space $\K$ and a projective quantum probability measure $\omega:\F(X)\rightarrow\B(\K)$ such that $\ost_\nu$ and $\ost_\omega$
are unitally completely order isomorphic;
\item\label{ap-cl-3} for all Hilbert spaces $\K$, all operators $L_1,\dots, L_m\in\B(\K)$, and every measurement basis $\B_\nu=\{A_1,\dots,A_m\}$ for $\nu$, the following two
properties hold:
\begin{enumerate}
\item\label{ap-cl-3a} $\B_\nu$ has trivial residual, and
\item\label{ap-cl-3b} $\left\| \displaystyle\sum_{j=1}^m A_j\otimes L_j\right\|\,=\,\max_{1\leq j\leq m}\|L_j\|$;
\end{enumerate}
\item\label{ap-cl-4} for every measurement basis $\B_\nu=\{A_1,\dots,A_m\}$ for $\nu$, 
the following two
properties hold:
\begin{enumerate}
\item\label{ap-cl-4a} $\B_\nu$ has trivial residual, and
\item\label{ap-cl-4b} $\lambda_{\rm max}(A_j)=1$ and $\lambda_{\rm min}(A_j)=0$ for every $j=1,\dots,m$;
\end{enumerate}
\item\label{ap-cl-5} for each measurement basis $\B_\nu=\{A_1,\dots,A_m\}$ for $\nu$, 
the following two
properties hold:
\begin{enumerate}
\item\label{ap-cl-5a} $\B_\nu$ has trivial residual, and
\item\label{ap-cl-5b} there exist a (possibly nonseparable) Hilbert space $\H_\varrho$, a faithful unital representation $\varrho:\B(\H)\rightarrow\B(\H_\varrho)$, 
a subspace $\H_0\subset\H_\varrho$, positive operators $Y_1,\dots,Y_m\in\B(\H_0)$, and pairwise-orthogonal projections 
$Q_1,\dots,Q_m\in\B(\H_0^\bot)$ such that, for every $j=1,\dots, m$,
$\|Y_j\|\leq1$  and $\varrho(A_j)=Q_j\oplus Y_j\in B(\H_0^\bot\oplus\H_0)$.
\end{enumerate}
\end{enumerate}
If, in addition, $\H$ has finite dimension, then the representation $\varrho$ in \eqref{ap-cl-5b} is the identity and 
the Hilbert space $\H_\varrho$ is $\H$ itself.
\end{theorem}

\begin{proof} \eqref{ap-cl-1} $\Rightarrow$ \eqref{ap-cl-2}. 
Suppose that $\{B_1,\dots,B_m\}$ is a linear
basis for $\T_\nu$, where $B_j=\nu(F_j)$ for some $F_j\in\F(X)$. By Naimark's Dilation Theorem there are a Hilbert space $\K$, an isometry $V:\H\rightarrow\K$, and a
projective measurement $\omega:\F(X)\rightarrow \B(\K)$ such that $\nu(E)=V^*\omega(E) V$ for all $E\in\F(X)$.
If $\Phi:\T(\H)\rightarrow\T(\K)$ denotes
the quantum channel $\Phi(R)=VRV^*$ and if $\phi=\Phi^*$, then Naimark's Dilation Theorem takes the form $\nu=\phi\circ\omega$. Because $\nu$ is approximately clean,
$\omega=\psi\circ\nu$ for an approximately normal ucp map $\psi:\B(\H)\rightarrow\B(\K)$. Thus $\T_\omega$ and $\T_\nu$ necessarily have the same dimension, namely $m$.

Let $Q_j=\omega(F_j)\in\T_\omega$ for each $j$. By Proposition \ref{basispartition}, there are pairwise disjoint sets $E_1,\dots, E_m\in\F(X)$ such that, if $P_j=\omega(E_j)$ for each $j$,
then $\{P_1,\dots,P_m\}$ is a measurement basis for $\T_\omega$ with trivial residual. For each $j$ let $A_j=\nu(E_j)$ and note that $P_j=\psi(A_j)$. Thus, $\{A_1,\dots,A_m\}$
is a measurement basis for $\T_\nu$ with trivial residual.
 
Moreover, 
by Lemma \ref{upper bound}, for every $p\in\mathbb N$ and all $L_1,\dots, L_m\in\M_p(\mathbb C)$, we have
\[
\begin{array}{rcl}
\displaystyle\max_{1\leq j\leq m}\|L_j\| &=& \left\|\displaystyle\sum_{j=1}^m P_j\otimes L_j\right\| \,=\,\left\|\displaystyle\sum_{j=1}^m \psi(A_j)\otimes L_j\right\| \\ && \\
&\leq& \left\|\displaystyle\sum_{j=1}^m A_j\otimes L_j\right\| \\&&\\
&\leq& \displaystyle\max_{1\leq j\leq m}\|L_j\| \,.
\end{array}
\]
Thus, $\psi_{\vert\ost_\nu}$ is a unital completely contractive map of $\ost_\nu$ onto $\ost_\omega$, and so $\psi_{\vert\ost_\nu}$ is a unital complete order isomorphism
\cite[Proposition 2.11]{Paulsen-book}.

\eqref{ap-cl-2} $\Rightarrow$ \eqref{ap-cl-3}. Fix a measurement basis $\B_\nu=\{A_1,\dots,A_m\}$ for $\nu$. 
By assumption, there are a Hilbert space $\K$ and a 
projective quantum probability measure $\omega:\F(X)\rightarrow\B(\K)$ such that $\ost_\nu$ and $\ost_\omega$
are unitally completely order isomorphic. Because every measurement basis for $\ost_\omega$ has trivial residual, the same is true of $\T_\nu$.
If $\psi:\ost_\nu\rightarrow\ost_\omega$ denotes this complete order isomorphism, then for any finite-dimensional Hilbert
space $\K$ and $L_1,\dots, L_m\in\B(\K)$, we have, by Lemma \ref{upper bound}, that
\[
\left\| \sum_{j=1}^m A_j\otimes L_j\right\|\,=\,\left\| \sum_{j=1}^m \psi(A_j)\otimes L_j\right\|\,=\,\max_{1\leq j\leq m}\|L_j\|\,.
\]
Now if $T$ is a bounded linear operator on a Hilbert space $\mathcal L$ of dimension at least $2$, then for every $\varepsilon>0$ there is a unit vector $\xi\in\mathcal L$ 
such that $\|T\|<\|T\xi\|+\varepsilon$. If $\mathcal L_\xi=\mbox{Span}\{\xi, T\xi\}$, and if $Q\in\B(\mathcal L)$ is the projection 
with range $\mathcal L_\xi$, then $\|QTQ\|<\|T\xi\|+\varepsilon$.
That is, to know the norm of an operator $T$, it is sufficient to know the norms of all compressions $QTQ$ of $T$ using projections of rank $2$. Hence, if the norm equalities above
hold for all finite-dimensional Hilbert spaces $\K$, then they also hold for all arbitrary Hilbert spaces $\K$.

\eqref{ap-cl-3} $\Rightarrow$ \eqref{ap-cl-4}. Let $\B_\nu=\{A_1,\dots,A_m\}$ be a measurement basis for $\nu$.
Fix $i$ and let $L_i=1$ and $L_j=0$ for $j\neq i$. Then, by assumption, $\|A_i\otimes L_i\|=\|L_i\|$ and so $\|A_i\|=1$. Since
$\lambda_{\rm max}(A_i)=\|A_i\|$, we deduce that $\lambda_{\rm max}(A_i)=1$. Because every spectral point of a hermitian operator is an approximate eigenvalue, the fact that
$1$ is in the spectrum of $A_i$ implies that there is a sequence of unit vectors $\xi_n\in\H$ such that $\lim_n\|A_i\xi_n-\xi_n\|=0$. Because, $A_i-1=-\sum_{j\neq i}A_j$ we 
rewrite this limit as $\lim_n\|\sum_{j\neq i}A_j\xi_n\|=0$. Hence, by the Cauchy--Schwarz inequality,
\[
\lim_{n\rightarrow\infty}\left\langle\sum_{j\neq i}  A_j\xi_n,\xi_n\right\rangle \,=\,0\,.
\]
That is,
\[
0\,=\,\lim_{n\rightarrow\infty}\sum_{j\neq i} \langle A_j^{1/2}\xi_n,A_j^{1/2}\xi_n\rangle \,=\, 
\sum_{j\neq i} \left(\lim_{n\rightarrow\infty}\|A_j^{1/2}\xi_n\|^2\right)\,.
\]
and so $0=\lim_{n\rightarrow\infty}\|A_j^{1/2}\xi_n\|$ for each $j\neq i$. Thus, $0$ in the spectrum of $A_j$ for all $j\neq i$. That is, $\lambda_{\rm min}(A_j)=0$ for all $j\neq i$.
Our choice of $i$ was arbitrary; thus, for every $j=1,\dots, m$ we necessarily have $\lambda_{\rm max}(A_j)=1$ and $\lambda_{\rm min}(A_j)=0$.

\eqref{ap-cl-4} $\Rightarrow$ \eqref{ap-cl-5}. Fix a measurement basis $\B_\nu=\{A_1,\dots,A_m\}$ for $\nu$.
Let $\varrho:\B(\H)\rightarrow\B(\H_\varrho)$ be the unital representation of $\B(\H)$ that is described in \cite{berberian1962}, which has
the property that $\lambda$ is an eigenvalue of $\varrho(Y)$ for every spectral point $\lambda$ of a hermitian operator $Y\in\B(\H)$. 
(Hence, $\H_\varrho$ is necessarily non separable if $\H$ has infinite dimension.) Therefore, by the hypothesis that $\lambda_{\rm max}(A_j)=1$ and $\lambda_{\rm min}(A_j)=0$ for each $j$,
we deduce that $0$ and $1$ are eigenvalues of the positive operators $\varrho(A_j)$ for all $j=1,\dots,m$. 

For simplicity of notation, we shall drop the reference to the unital homomorphism $\varrho$ and assume that $0$ and $1$ 
are eigenvalues of the positive operators $A_j$ for all $j=1,\dots,m$. (In the cases where $\dim\H$ is finite, we may dispense with $\varrho$ altogether because spectral points are necessarily
eigenvalues in such cases.)

Decompose $\H$ as $\H=\ker(A_1-1)\oplus \H_0^1$. Because both $A_j$ and $1-A_j$ are positive, and because $\sum_{j=1}^mA_j=1$, the
decomposition $\H=\ker(A_1-1)\oplus \H_1$ implies that
\[
A_1\,=\,\left[ \begin{array}{cc} 1& 0  \\ 0&G_1 \end{array}\right] \quad\mbox{and}\quad
A_j\,=\,\left[ \begin{array}{cc} 0& 0  \\ 0&G_j \end{array}\right]\,\;\mbox{for all }j\geq 2\,,
\]
for some positive contractions $G_1,\dots,G_m\in\B(\H_1)$ satisfying $\sum_{j=1}^mG_j=1$. Because $1$ is an eigenvalue of $A_2$, the matrix
representation of $A_2$ above shows that $\ker(A_2-1)$ must lie within the subspace $\H_1$. Thus, we decompose the Hilbert space $\H$ further as
$\H=\ker(A_1-1)\oplus\ker(A_2-1)\oplus\H_2$ so that
\[
A_1\,=\,\left[ \begin{array}{ccc} 1& 0 & 0 \\ 0&0&0 \\ 0&0& K_1 \end{array}\right],\quad
A_2\,=\,\left[ \begin{array}{ccc} 0& 0 & 0 \\ 0&1&0 \\ 0&0& K_2 \end{array}\right],\;\mbox{and}\;
A_j\,=\,\left[ \begin{array}{ccc} 0& 0 & 0 \\ 0&0&0 \\ 0&0& K_j \end{array}\right] 
\]
for all $j\geq3$ and
for some positive contractions $K_1,\dots,K_m\in\B(\H_2)$ satisfying $\sum_{j=1}^mK_j=1$.
As before, $A_3$ has an eigenvalue $1$; the corresponding eigenspace $\ker (A_3-1)$ must lie in $\H_3$, by the matrix representation of $A_3$ above.
Therefore, an iteration of the argument used to this point produces, after $m$ steps in total, a decomposition of $\H$ as 
\[
\H\,=\,\left(\bigoplus_{j-1}^m\ker(A_j-1)\right)\,\bigoplus\,\H_m\,,
\]
for a subspace $\H_m\subset\H$, and where, for every $j$,  $A_j=Q_j\oplus Y_j$ for some positive contractions $Y_j\in\B(\H_m)$ such that $\sum_{j=1}^mY_j=1\in\B(\H_m)$
and where each $Q_j$ is the projection with range $\ker(A_j-1)$.

\eqref{ap-cl-5} $\Rightarrow$ \eqref{ap-cl-1}. Suppose that $\nu'$ is a quantum probability measure and $\nu\lea\nu'$.
The proof of Proposition \ref{basispartition} shows that there exists pairwise-disjoint measurable sets $E_1,\dots,E_m\in\F(X)$ such that, if $A_j=\nu(E_j)$ for each $j$, then
$\{A_1,\dots, A_m\}$ is a measurement basis for $\nu$. By assumption \eqref{ap-cl-5}, 
there exist a (possibly nonseparable) Hilbert space $\H_\varrho$, a faithful unital representation $\varrho:\B(\H)\rightarrow\B(\H_\varrho)$, 
a subspace $\H_0\subset\H_\varrho$, positive operators $Y_1,\dots,Y_m\in\B(\H_0)$, and pairwise-orthogonal projections 
$Q_1,\dots,Q_m\in\B(\H_0^\bot)$ such that, for every $j=1,\dots, m$,
$\|Y_j\|\leq1$  and $\varrho(A_j)=Q_j\oplus Y_j\in B(\H_0^\bot\oplus\H_0)$. Therefore, for any finite-dimensional Hilbert space $\K$ and operators $L_1,\dots, L_m\in\B(\K)$,
\[
\begin{array}{rcl}
\left\| \displaystyle\sum_{j=1}^m A_j\otimes L_j\right\| &=& \left\| \displaystyle\sum_{j=1}^m \varrho(A_j)\otimes L_j\right\| \,=\,\left\| \displaystyle\sum_{j=1}^m (Q_j\oplus Y_j)\otimes L_j\right\| \\ && \\
&=& \max\left\{ \left\| \displaystyle\sum_{j=1}^m Q_j\otimes L_j\right\|,\;\left\| \displaystyle\sum_{j=1}^m Y_j\otimes L_j\right\| \right\} \\ && \\
&=& \displaystyle\max_{1\leq j\leq m} \| L_j\|  ,
\end{array}
\]
since, by Lemma \ref{upper bound}, 
\[
\left\| \displaystyle\sum_{j=1}^m Q_j\otimes L_j\right\|\,=\,\displaystyle\max_{1\leq j\leq m} \| L_j\|  \quad\mbox{and}\quad
\left\| \displaystyle\sum_{j=1}^m Y_j\otimes L_j\right\| \,\leq\, \displaystyle\max_{1\leq j\leq m} \| L_j\| \,.
\]
Therefore, by Lemma \ref{upper bound} again,
\[
\left\|  \sum_{j=1}^m \nu'(E_j)\otimes L_j\right\|\, \leq\,
 \displaystyle\max_{1\leq j\leq m} \| L_j\| 
 \,=\, \left\|  \sum_{j=1}^m \nu(E_j)\otimes L_j\right\|\,.
\]
Hence, Theorem \ref{fd1} implies that $\nu'\lea\nu$.
 \end{proof}
 
\section{Observations and Applications}\label{S:remarks}

\subsection{Approximately clean 1-0 measurements}

A \emph{1-0 measurement} is one in which the sample space of outcomes is given by $X=\{0,1\}$ and the space of events by
$\F(X)=\{\emptyset, \{0\},\{1\}, X\}$. 
A direct application of Theorem \ref{fd2} leads to the following
equivalent statements for a quantum probability measures $\nu$ on $(X,\F(X))$:
\begin{enumerate}
\item\label{two-a} $\nu$ is approximately clean;
\item\label{two-b} $\lambda_{\rm min}(\nu(\{0\}))=0$ and $\lambda_{\rm max}(\nu(\{0\}))=1$.
\end{enumerate}

This spectral condition above also appears in another notion of optimality for quantum measurements.  
By the notation $\nu_2 \lef \nu_1$ one means that $\nu_2$ is a \emph{fuzzy version} of $\nu_1$, which is to say that there is a confidence mapping $\Lambda$
on $\F(X)$ such that $\tr\left(R\nu_2(E)\right)=\tr\left(R\nu_1(\Lambda(E))\right)$ for all measurable sets $E$ and density operators $R$ \cite{heinonen2005}.
With respect to this partial order,
$\nu$ is optimal if $\nu \lef \nu'$ implies $\nu' \lef \nu$. In the case where
$X=\{0,1\}$ and $\F(X)$ is the power set of $X$, a 1-0 quantum measurement $\nu$ is optimal with respect to $\lef$ if and only if $\nu(\{0\})$
satisfies the spectral property \eqref{two-b} 
above \cite[Proposition 3]{heinonen2005}.
Hence, we deduce that a 1-0 measurement is approximately clean if and only if it is optimal with respect to the fuzzy ordering $\lef$.

\subsection{Approximately clean but not clean}

\begin{proposition}\label{drf} Assume that $\H$ is a separable Hilbert spaces of infinite dimension and that $X=\{0,1\}$ is an outcome space
for 1-0 measurements.
\begin{enumerate}
\item\label{drf-1} There exists an approximately normal ucp map $\phi:\B(\H)\rightarrow\B(\H)$ such that $\phi$
is not normal.
\item\label{drf-2} There exist approximately clean 1-0 measurements $\omega, \nu:\F(X)\rightarrow\B(\H)$ such that $\omega \lea \nu$, but $\omega\not\lec\nu$.
\item\label{drf-3} There exist approximately clean 1-0 measurements $\omega, \nu:\F(X)\rightarrow\B(\H)$ such that $\nu \lec \omega$, but $\omega\not\lec\nu$.
\end{enumerate}
\end{proposition}

\begin{proof} To prove statement \eqref{drf-1}, assume without loss of generality that $\H=L^2([0,1])$
and suppose that $A\in\B(\H)$ is the multiplication operator $Af(t)=tf(t)$, $f\in\H$, and let $P\in\B(\H)$ be a projection
such that $P$ and $1-P\neq 0$ are of infinite rank. Thus, $A$ and $P$
are positive and $\lambda_{\rm min}(A)=\lambda_{\rm min}(P)=0$ and $\lambda_{\rm max}(A)=\lambda_{\rm max}(P)=1$.
Therefore, by Theorem \ref{fd1} and its proof,
there is an approximately
normal ucp map $\phi:\B(\H)\rightarrow\B(\H)$
for which $P=\phi(A)$. Suppose, further, that $\phi$ is normal.
Thus, $P=\sum_k W_k^*AW_k$ for some sequence $\{W_k\}_{k}\subset\B(\H)$ of contractions. Choose any nonzero $\xi\in\H$ such that $P\xi=0$; thus,
\[
0\,=\,\langle P\xi,\xi\rangle\,=\,\sum_k \langle AW_k\xi,W_k\xi\rangle\,=\,\sum_k \|A^{1/2}W_k\xi\|^2\,.
\]
Therefore, $AW_k\xi=0$ for every $k$. Because $A$ has no eigenvalues, necessarily $W_k\xi=0$ for every $k$, and so also $W_k^*W_k\xi=0$ for every $k$. Thus,
$0\neq\xi=\sum_k W_k^*W_k\xi=0$, which is a contradiction. Thus, there cannot be a normal ucp map $\phi$ for which $P=\phi(A)$. Hence, $\phi$ is approximately
normal but not normal.

For statement \eqref{drf-2}, let $\omega,\nu:\F(X)\rightarrow\B(\H)$ denote the unique 1-0 quantum measurements for which $\omega(\{0\})=P$
and $\nu(\{0\})=A$. The proof of assertion \eqref{drf-1} shows that there is an approximately normal $\phi$ such that $\omega=\phi\circ\nu$ but there are no normal ucp $\psi$
satisfying $\omega=\psi\circ\nu$. Thus, $\omega \lea \nu$, but $\omega\not\lec\nu$.

To prove statement \eqref{drf-3}, let us use $\omega$ and $\nu$ as above. Because \eqref{drf-2} already asserts that $\omega\not\lec\nu$, we need only show that 
$\nu \lec \omega$. Theorem \ref{fd2}\eqref{ap-cl-3} implies that $\|A\otimes L + 1\otimes K\|=\|P\otimes L+1\otimes K\|$ for all $n\times n$ matrices $L,K$ and all $n\in\mathbb N$.
Hence, by \cite[Theorem 2.4.2(iv)]{arveson1972}, there is a separable
Hilbert space $\K$, a homomorphism $\pi:\cstar(P)\rightarrow\B(\K)$ (where $\cstar(P)$ is the unital C$^*$-algebra generated by $P$) and an isometry $V:\H\rightarrow\K$ such that
$A=V^*\pi(P)V$. Because $A$ has infinite and $1-A$ have infinite rank, so do $\pi(P)$ and $1_{\K}-\pi(P)$. Thus, there is a unitary operator $U:\H\rightarrow\K$ such that $\pi(P)U=UP$.
Define a normal ucp map $\psi:\B(\H)\rightarrow\B(\H)$ by $\psi(X)=(U^*V)^*X(U^*V)$ and observe that $A=\psi(P)$. That is, $\nu=\psi\circ\omega$ for a normal ucp map $\psi$, 
and therefore $\nu\lec\omega$.
\end{proof}

\begin{corollary} There exists a 1-0 quantum measurement that is approximately clean but not clean.
\end{corollary}

\begin{proof} Let $\omega$ and $\nu$ be the approximately clean 1-0 measurements discussed in the proof of Proposition \ref{drf}. Because
$\nu \lec \omega$ and $\omega \not\lec\nu$, the approximately clean quantum probability measure $\nu$ fails to be clean.
\end{proof}

\subsection{Clean qubit measurements are projective}

If $\H$ is finite-dimensional, then a quantum measure is clean if and only if it is approximately clean. Furthermore, if $\dim\H=2$, then
Theorem \ref{fd2} implies that 
$\nu$ is clean if and only if $\nu$ is a projection-valued
measure. In contrast, a qubit measurement $\nu$ with sample space $X=\{x_1,\dots,x_n\}$
is clean in the original sense of \cite{clean2005,kahn2007}
if and only if $\nu(\{x_j\})$ has rank-1 for every $j=1,\dots,n$ \cite[Theorem 11.2]{clean2005}, \cite{kahn2007}. Thus, the stricter criteria for cleanness
used herein leads to a smaller class of clean quantum measurements than was found in \cite{clean2005,kahn2007}.

\subsection{Structure of quantum probability measures with finite-dimensional measurements spaces}

\begin{proposition}\label{structure} If $\{A_1,\dots,A_m\}$ is a measurement basis for a quantum probability measure $\nu$, 
then there exist finite signed measures
$\upsilon_1$, \dots, $\upsilon_m$ on $(X,\F(X))$ such that each  $\upsilon_j\leac\nu$ 
(in the sense that $\nu(E)=0$ implies $\upsilon_j(E)=0$) and
\[
\nu(E)\,=\,\sum_{j=1}^m \upsilon_j(E)A_j\,,\;\mbox{for all }E\in\F(X)\,.
\]
\end{proposition}

\begin{proof}  For every $E\in\F(X)$ there exist unique $\alpha_1^E, \dots,\alpha_m^E\in\mathbb R$ such that $\nu(E)=\sum_{j=1}^m\alpha_j^E A_j$.
Define $\upsilon_j:\F(X)\rightarrow\mathbb R$ by $\upsilon_j(E)=\alpha_j^E$, $j=1,\dots,m$. By the linear independence of $A_1,\dots, A_m$, $\upsilon_j(\emptyset)=0$
for each $j$, while $\upsilon_j(X)=1$ for all $j$. To show that $\upsilon_j$ is countably additive, suppose that $\{E_k\}_{k\in\mathbb N}$
is a countable collection of pairwise disjoint measurable sets. By the countable additivity of $\nu$, 
\[
\nu\left(\displaystyle\bigcup_{k\in\mathbb N}E_k\right) \,=\, \displaystyle\sum_{k\in\mathbb N}\nu(E_k)
\,=\,\displaystyle\sum_{k\in\mathbb N}\sum_{j=1}^m\alpha_j^{E_k}A_j \,=\,
\displaystyle\sum_{j=1}^m\left(\displaystyle\sum_{k\in\mathbb N}\alpha_j^{E_k}\right)A_j\,,
\]
which proves---again by the linear independence of $A_1,\dots, A_m$---that 
\[
\upsilon_j\left(\displaystyle\bigcup_{k\in\mathbb N}E_k\right) \,=\, \displaystyle\sum_{k\in\mathbb N}\upsilon_j(E_k)\,,
\]
for every $j=1,\dots,m$. Lastly, if $\nu(E)=0$, then the linear independence of $A_1$, \dots, $A_m$ yields $\upsilon_j(E)=0$ for every $j$,
whence each $\upsilon_j\leac\nu$.
\end{proof}

In Proposition \ref{structure}, it is not generally true that the signed measures $\upsilon_j$ are in fact positive measures. 

\begin{definition} A measurement basis $\B_\nu=\{A_1,\dots,A_m\}$ for a quantum probability measure $\nu$ is \emph{perfect} if there exist (positive) probability measures
$\mu_1$, \dots, $\mu_m$ on $(X,\F(X))$ such that
\[
\nu(E)\,=\,\sum_{j=1}^m \mu_j(E)A_j\,,\;\mbox{for all }E\in\F(X)\,.
\]
\end{definition}

It is not difficult to verify that every measurement basis $\B_\omega$ of a projective quantum probability measure $\omega$ is perfect.

In quantum information theory it is a standard practice to define
a POVM associated with an $n$-outcome quantum measurement to be an $n$-tuples of positive operators
$M_1,\dots, M_n$ such that $\sum_{j=1}^nM_j=1$. In this scenario, the underlying quantum probability measure $\nu:\F(X)\rightarrow\B(\H)$ is the unique POVM for which
$\nu(\{x_j\})=M_j$ for each $j$, where $X=\{x_1,\dots,x_n\}$. Observe that $\F(X)$ is the power set of $X$ and that
\[
\R_\nu\,=\,\left\{\sum_{j\in E}M_j\,:\,E\in\F(X)\right\}\,\subset\,\mbox{Span}\left\{M_1,\dots,M_n\right\}\,.
\]
Thus, $\mbox{Span}\{M_1,\dots,M_n\}=\ost_\nu$. However, any measurement basis $\{A_1,\dots,A_m\}$ 
for $\ost_\nu$ is drawn from the set $\R_\nu$ rather than from $\{M_1,\dots,M_n\}$. This is an essential difference in our approach from other studies, where the POVM
is analysed by studying the spanning set $\{M_1,\dots,M_n\}$. By way of the observation above it is clear that $\dim\ost_\nu\leq |X|$, the cardinality of the sample space $X$.

 \begin{proposition} If $\{M_1,\dots,M_n\}$ is a POVM associated with a quantum probability measure $\nu$ on a sample space $X$ of cardinality $n$, then the following
statements are equivalent:
\begin{enumerate}
\item $\dim\ost_\nu=n$;
\item $\{M_1,\dots,M_n\}$ is a perfect measurement basis for $\nu$.
\end{enumerate}
\end{proposition}

\begin{proof} If $\dim\ost_\nu=n$, then $\mbox{Span}\{M_1,\dots,M_n\}$ is $n$-dimensional, which implies that $\{M_1,\dots,M_n\}$ is a basis of $\ost_\nu$. Moreover, 
$\{M_1,\dots,M_n\}$ is clearly a measurement basis and
\[
\nu\,=\,\sum_{j=1}^n \delta_{\{x_j\}}M_j\,,
\]
where $\delta_{\{x_j\}}$ is the Dirac probability measure with mass concentrated at the point set $\{x_j\}$. 
The converse is obvious.
\end{proof}

\subsection{Quantity of information versus quality of information}

It was noted by Buscemi \emph{et al} in \cite{clean2005} that in passing to cleaner quantum measurements there is a trade off
between the quantity of information and quality of information afforded by such measurements. As a specific example, we indicate below
that quantum measurements that yield the greatest amount of information are never approximately clean.

Recall that a quantum probability measure $\nu$ on $(X,\F(X))$ with values in $\B(\H)$ is \emph{informationally complete} if,
for any two states $R_1,R_2\in\T(\H)$, the equation
\[
\tr\left(R_1\nu(E)\right)\,=\,\tr\left(R_2\nu(E)\right) 
\]
holds for every event $E\in\F(X)$ 
if and only if $R_1=R_2$.

\begin{proposition} 
No informationally complete quantum measurement with values in
$\B(\H)$, where $\H$ is a Hilbert space of finite dimension $d\geq2$,
is clean.
\end{proposition}

\begin{proof} It is well-known that if $\nu$ is informationally complete, then the only operators that commute with every effect $\nu(E)$, $E\in\F(X)$, 
are the scalar operators. On the other hand, statement \eqref{ap-cl-5b} of Theorem \ref{fd2} asserts that if $\nu$ is clean and $\dim\H\geq2$, then there is a non-trivial projection
in the commutant of $\T_\nu$.
\end{proof}

\section{Conclusion}

The notion of cleanness for a quantum probability measure is one way in which a positive operator valued measure can be viewed as being optimal. Our work herein
is strongly motivated by the paper of Buscemi \emph{et al} \cite{clean2005} in the context of finite-dimensional spaces. In particular, we replace their condition that
the underlying Hilbert space have finite dimension with the weaker condition that the measurement space associated with a quantum probability measure have
finite dimension, which is of practical interest for the measurement of quantum systems represented by an infinite-dimensional Hilbert space in which only finite many 
outcomes of the measurement are expected. Even in this case, however, the use of infinite-dimensional spaces leads some further challenges, one of which being that the 
space of normal ucp maps is not closed in the point-ultraweak topology. For this reason, our use
of Pellonp{\"a}{\"a}'s refinement of the notion of cleanness \cite{pellonpaa2011} involves approximately normal rather than normal ucp maps. Although there is a true
difference between clean and approximately clean in the contact of infinite-dimensional spaces, the latter notion allows for a complete characterisation (Theorem \ref{fd2}) 
of such quantum probability measures when the measurement space has finite dimension. In the case of finite-dimensional Hilbert space, Theorem \ref{fd2} also represents
a new result, as previous work in finite dimensions did not take Pellonp{\"a}{\"a}'s definition of cleanness into account.
 
The notion of projective quantum probability measure dates back to von Neumann \cite{von Neumann-book1932}
and is associated with nondestructive measurements. Such quantum measures are clean, as
Pellonp{\"a}{\"a} \cite{pellonpaa2011} demonstrates.  Herein we proved in Theorem \ref{fd2} that a non-projective
quantum probability measure $\nu$ with finite-dimensional measurement space $\T_\nu$ is approximately clean if and only if there is a projective measure $\omega$
such that $\T_\nu$ and $\T_\omega$ are completely order isomorphic. Thus, although the measurement outcomes of $\nu$ are not registered as projections, the resulting
measurement space is, from the perspective of linearity and matricial
order, exactly the same as that space $\T_\omega$ which does arise from measurement outcomes that are registered as projections. Therefore, in the category of finite-dimensional operator systems \cite{Paulsen-book}, 
approximately clean and projective measurements determine the same objects.

In practice, quantum measurements are rarely nondestructive, and so the introduction of approximately clean measurements allows for the opportunity to model
(nearly) nondestructive measurements while at the same time accommodating within the mathematical model the very real presence of noise in the measuring apparatus. Indeed,
in restricting to quantum systems modelled by finite-dimensional Hilbert spaces, Theorem \ref{fd2} demonstrates that a clean quantum measurement $\nu$ has a component that is projective
and an orthogonal component that is noisy, and that it is this projective component that is responsible for the cleanness of $\nu$. Consequently, all clean qubit measurements are projective.
Some strategies for the preprocessing of noisy quantum measurements to achieve cleaner ones have been considered recently in \cite{purification2010}.

\section*{Acknowlegements}

This work is supported in by the Discovery Grant and Postgraduate Scholarship programs of
the Natural Sciences and Engineering Research Council of Canada. 
We wish to thank the referee for a
useful critique of the initial version of this paper and for drawing our attention to the references \cite{purification2010,heinonen2005}.
 

\end{document}